\documentclass{article}
\usepackage{amssymb,amsmath}
\usepackage[numbers,sort&compress,longnamesfirst,sectionbib]{natbib}
\usepackage[linktocpage,
            colorlinks=true,linkcolor=black,
            citecolor=black,urlcolor=black,
            pdfpagemode=UseNone,
            bookmarks=true,
            bookmarksopen=false,
            hyperfootnotes=false,
            pdfhighlight=/I,
            ]{hyperref}
\usepackage{fancyhdr}
\usepackage[algo2e,vlined,linesnumbered,ruled,titlenotnumbered]{algorithm2e}
\usepackage{algorithmic}
\usepackage{float}
\floatstyle{plain}
\newfloat{tobbifloat}{thp}{lop}
\floatname{tobbifloat}{Algorithm}
  
\usepackage{amsthm}
\newtheorem{theorem}{Theorem}
\newtheorem{lemma}[theorem]{Lemma}

\newtheorem{definition}[theorem]{Definition}

\newcommand{\Oh}{\mathcal{O}}
\newcommand{\cG}{{\ensuremath{\mathcal{G}}}}
\newcommand{\A}{{\mathcal{A}}}
\newcommand{\B}{{\mathcal{B}}}
\newcommand{\R}{{\mathcal{R}}}
\newcommand{\N}{\mathbb{N}}
\newcommand{\Prob} {\mathsf{Pr}}
\usepackage{xspace}
\newcommand{\K}{{\mathcal{K}}}
\newcommand{\E} {\mathsf{E}}
\newcommand{\eps}{\varepsilon}
\newcommand{\kClique} {\textsc{$k$\nobreakdash-Clique}\xspace}

\usepackage{graphicx}
\usepackage{setspace,color}   
\begin{document}

\title{\textbf{On the Average-case Complexity\\of Parameterized Clique}}

\author{Nikolaos Fountoulakis \and Tobias Friedrich \and Danny Hermelin}
\date{}
\maketitle

\begin{abstract}
The \kClique problem is a fundamental combinatorial problem that plays a prominent role in classical as well as in parameterized complexity theory. It is among the most well-known {\sc NP}-complete and {\sc W[1]}-complete problems. Moreover, its average-case complexity analysis has created a long thread of research already since the 1970s. Here, we continue this line of research by studying the dependence of the average-case complexity of the \kClique problem on the parameter $k$. To this end, we define two natural parameterized analogs of efficient average-case algorithms. We then show that \kClique admits both analogues for Erd\H{o}s-R\'{e}nyi random graphs of \emph{arbitrary} density. We also show that \kClique is unlikely to admit neither of these analogs for some specific computable input distribution.
\end{abstract}

\section{Introduction}
\label{Section: Introduction}

The \kClique problem is one of the most fundamental combinatorial problems in graph theory and computer science. This problem asks to determine whether a given graph contains a clique of size $k$, \emph{i.e} a complete subgraph on $k$ vertices. The \kClique problem forms the groundwork for many worst-case hardness frameworks: It is one of Karp's famous initial list of NP-complete problems~\cite{Karp72}, and its optimization variant is a classical example of a problem
that is NP-hard to approximate within a factor of $n^{1-\varepsilon}$ for any $\varepsilon > 0$~\cite{Zuckerman07}. In parameterized complexity theory~\cite{DowneyFellows1999}, the \kClique problem is textbook example complete for the class W[1], the parameterized analog of NP, playing a prominent role in W[1]-hardness results very much akin to the role \textsc{3-SAT} plays in the classical complexity.

In this paper we are interested in the parameterized complexity of the \kClique problem on ``average" inputs. For our purposes, an average \kClique instance can be naturally and conveniently modeled using the thoroughly-studied Erd\H{o}s-R\'enyi distributions on graphs. The class of these distributions is typically denoted by $\cG (n,p)$, with $n \in \mathbb{N}$ and $p \in [0,1]$, where on a graph with $n$ vertices each pair of vertices are adjacent independently with probability $p$. Such random graphs have approximate density $p$, and it is well-known (see \emph{e.g.}~\cite{Bol,JLR}) that the typical properties of these random graphs are essentially the typical properties of a random graph that is uniformly selected among all graphs on $n$ vertices and $p{n \choose 2}$~edges.

The question of of finding cliques in $\cG(n,p)$ random graphs has been raised by Karp~\cite{Karp76} already in 1976. Karp observed that in $\cG (n,1/2)$ (note that this is in fact the uniform distribution over all graphs on $n$ vertices) the maximum size of a clique is about $2\log n$ with high probability, but the greedy algorithm only finds with high probability a clique that is approximately half this size. Karp asked whether in fact there is any polynomial-time algorithm that finds a clique of size $(1+\eps)\log n$, for some $\eps >0$. This question remains open until today.

Finding cliques in $\cG(n,p)$ random has also been considered when the clique sought after have small size, which is the main theme of our paper. For a fixed integer $k\geq 3$, the random graph $\cG (n,p)$ undergoes a phase transition regarding the (almost sure) existence  of cliques of size $k$ (cf.~\cite{Bol}
or~\cite{JLR}) as the edge probability $p$ grows. More specifically, it is known that when $p \ll n^{-2/(k-1)}$, then $\cG (n,p)$ does not contain any cliques of size $k$, with high probability, but when $p\gg n^{-2/(k-1)}$, then in fact there are many $k$-cliques with high probability. However, inside the ``critical window", that is when $p = \Theta (n^{-2/(k-1)})$, the maximum size of a clique could be either $k-1$ or $k$ each one occurring with probability that is bounded away from 0 as $n$ grows to infinity. More precisely, the number of cliques of size~$k$ follows asymptotically a Poisson distribution with parameter that depends on~$k$. In this range, the greedy algorithm finds a clique of size $\lfloor {k\over 2} \rfloor$ or $\lceil {k \over 2} \rceil$, with high probability. Repeating the greedy algorithm $n^{\eps^2 k + O(1)}$ times, one can find a clique of size approximately $\left({1\over 2} + \eps\right)k$ with high probability (cf.~\cite{Rossman10}). Thus, taking $\eps = 1/2$, there is an algorithm that effectively finds all cliques in $\cG (n,p)$ which operates within time $n^{k/4 + O(1)}$ with high probability.

Since the above algorithm is the fastest algorithm known, it seems that a typical instance of $\cG (n,p)$ with $p = \Theta (n^{-2/(k-1)})$ is in fact a hard instance for \kClique. This is also suggested by the lower bounds on the size of monotone circuits for \kClique derived recently by~\citet{Rossman10} (see also \cite{Rossman2010}) for $p$ in this range. Thus any substantial improvement to the $n^{k/4 + O(1)}$ algorithm above would be a major breakthrough result; not to mention an FPT algorithm running in $f(k) \cdot n^{O(1)}$ time, which is perhaps far too much of an improvement than we can expect\footnote{Note that $f(k) \cdot n^{O(1)}$ $<\!\!\!<$ $n^k$ for any function $f$, when $k$ is fixed and $n$ tends to infinity.}. To avoid this obstacle, we consider distributions $\cG(n,p)$ where $p$ does not depend on $k$ (but may depend on $n$). Apart from the obvious advantage that this gives a real chance at obtaining positive results, we also
believe that this a very natural model of practical settings. Indeed, in many cases the distribution of the graphs we are interested in is fixed, while the size of the cliques we are looking for may vary.

We consider two types of algorithms running in FPT time on average. The first is an avgFPT-algorithm, which is an algorithm with expected $f(k) \cdot n^{O(1)}$ run-time. Thus, an avgFPT-algorithm is required to run in FPT-time on average according to the given input distribution. This means that the algorithm is allowed to be slow on some instances, so long as that its efficient on average.  The notion of avgFPT-time is a natural parameterized analog of an avgP-time algorithm (see \emph{e.g.}~\cite{Goldreich}), and is perhaps the most natural definition of the notion ``FPT on average''.

We present a very simple avgFPT algorithm for \kClique for essentially all distributions $p:=p(n)$. By essentially, we mean all \emph{natural} distributions that have typical properties, such as certain limit properties (this is made precise in Definition~\ref{def:natural p}). The first result of this paper is thus the following theorem.
\begin{theorem}
\label{Theorem: Clique is easy}Let $p:=p(n)$ denote a natural distribution function. There is an \textnormal{avgFPT}-algorithm for \kClique on graphs $G\in\cG(n,p)$.
\end{theorem}

The second type of average-case FPT algorithms we consider are algorithms that run in typical FPT (typFPT) time. By this we mean a running time of $f(k) \cdot n^{O(1)}$ with high probability, where high probability means that the algorithm is allowed to be slower only with probability smaller than any polynomial in~$n$. Thus, one may view the difference between a typFPT-time algorithm and an avgFPT-time algorithm is that an avgFPT-time algorithm is allowed to be slightly slow on relatively many instances, while a typFPT-time algorithm is allowed to be extremely slow on relatively few instances. In stochastic terms, this is precisely the difference between bounding the expected value of a random variable and showing that it is bounded with high probability. Again, the analogous notion in classical complexity is typical P-time~\cite{Goldreich}.

We show that the same algorithm used in Theorem~\ref{Theorem: Clique is easy} is actually a typFPT algorithm for \kClique for any natural $p:=p(n)$. However, the proof of this result is more involved than the former and requires a rather sophisticated tail bound argument.
\begin{theorem}
\label{Theorem: Clique is easy2}Let $p:=p(n)$ denote a natural distribution function. There is a \textnormal{typFPT}-algorithm for \kClique on graphs $G\in\cG(n,p)$.
\end{theorem}
It is worth mentioning that in both theorems above, our algorithms are completely deterministic and \emph{always} correctly decide whether their input graph contains a clique of size $k$. This makes the proofs more challenging, since the algorithms cannot only assume that a $k$-clique is unlikely to exist in the input, but they must also certify this somehow. Furthermore, our algorithms can easily be modified to determining whether a $\cG(n,p)$ random graph has an
independent set of size $k$. 
Moritz M{\"{u}}ller's PhD thesis~\cite{Muller2008} provides the first attempt at  setting up a framework of parameterized average case complexity. In particular, he defined a notion very much similar to our \textnormal{avgFPT}-algorithm, except that in his case the algorithm is allowed to have one-sides errors with constant probability. The notion of typFPT has not appeared elsewhere to the best of our knowledge. The distinction between these two types of average-case tractability notions is standard in the classical world, and in Section~\ref{Section: Average Case Parameterized Complexity Classes} we briefly argue why this distinction makes even more sense in the parameterized world. M{\"{u}}ller also defined an average-case analog of W[1], and showed that there is some (artificial) problem which is complete for it. We discuss this result in the last part of the paper, and show that the \kClique problem is hard for this average-case analog of W[1] on a specific distribution.

\section{Average Case Parameterized Algorithms}
\label{Section: Average Case Parameterized Complexity Classes}

In this section we define our two average-case analogs of FPT
algorithms. We begin by some necessary terminology which
follows the terminology used in~\citet{Goldreich} for classical
average-case analysis. A \emph{distribution ensemble} $X$ is an
infinite sequence of probability spaces, one for each $n \in
\mathbb{N}$, such that the $n$\nobreakdash-th space is defined
over $\{0,1\}^n$. We will associate with $X$ a sequence of
random variables $\{X_n\}^\infty_{n=1}$, where $X_n$ is
assigned strings in $\{0,1\}^n$ according to the corresponding
distribution in $X$ (thus, formally $X_n$ maps strings from
$\{0,1\}^n$ to strings from $\{0,1\}^n$). For example, we will
write $\textsf{Pr}[X_n=x]$ for the probability that $X_n$
equals a specific $x \in \{0,1\}^*$ when drawn at random
according to $X$. A \emph{distributional parameterized problem}
is a pair $(L,X)$, where $L \subseteq \{0,1\}^* \times
\mathbb{N}$ is a parameterized problem, and $X$ is a
distribution ensemble over strings in $\{0,1\}^*$.

Next let us consider avgFPT-time algorithms. Informally, we
would like this class of algorithms to contain all algorithms
running in FPT\nobreakdash-time on average according to the
distribution of their inputs. However, similar to the classical
world, there are some technical problems with simply requiring
that the corresponding algorithms run in expected
FPT\nobreakdash-time (\emph{e.g.} this does not allow for
robustness in the computation model, see~\cite{Goldreich}).
Thus, as is done in the classical setting, we will require some
sort of \emph{normalized} expected running time. Furthermore,
we require that our algorithms always output the correct
solution, or in other words, they must be able to decide the
given problem.

\begin{definition}
\label{Definition: avgFPT}Let $(L,X)$ be a distributional parameterized problem. We say
that an algorithm $\A$ deciding $L$ runs in
\textnormal{avgFPT}-time if there exists a constant~$c$ and a
function $f\colon\mathbb{N} \rightarrow \mathbb{N}$ such that
for all $k\in \mathbb{N}$:
\begin{equation*}
\label{eq:avgFPT}
\sum_{n\in \mathbb{N}} \mathsf{E}
\left[{t_{\A} (X_n,k) \over n^c}\right]
< f(k).
\end{equation*}
Here, and elsewhere, the random variable $t_\A(X_n,k)$ denotes
the running time of an algorithm $\A$ on input $(x,k)$, where
$x$ is chosen with probability
$\textnormal{\textsf{Pr}}[X_n=x]$.
\end{definition}

Observe that an avgFPT\nobreakdash-time algorithm may run the
brute-force procedure, which typically runs in $\Oh(n^k)$ time,
with probability $n^{-k}$. This, as we will see further on,
allows for a very simple analysis in some cases. A more
stringent requirement of an efficient algorithm for
parameterized distributional problems is to insist that it
\emph{typically} runs in FPT\nobreakdash-time. That is, that it
runs in FPT\nobreakdash-time with high probability, where high
probability means that the algorithm is allowed to be too slow
only with probability super-polynomially small. Thus, a
probability of $n^{-k}$ will not suffice. This indicates that
the distinction between the two average-case classes might be
more apparent in the parameterized world than it is in
classical complexity theory.

\begin{definition}
\label{Definition: typFPT}Let $(L,X)$ be a distributional parameterized problem. We say
that an algorithm $\A$ deciding $L$ runs in
\textnormal{typFPT}-time if there exists a function~$f$ and a
polynomial~$p$, such that for all $k\in \mathbb{N}$ and
polynomials~$q$ there is an $n_0\in\mathbb{N}$ such that
for all $n>n_0$:
\begin{equation*}
\label{eq:typFPT}
\Pr[t_\A(X_n, k) > f(k) \cdot p(n)] < \frac{1}{q(n)}.
\end{equation*}
\end{definition}

It is important to note that in the probability bound of the definition above we can equivalently use $f(k)/q(n)$ instead of $1/q(n)$. It is obvious that a $1/q(n)$ bound implies a $f(k)/q(n)$ bound (for $f(k)\geq1$).
To see the opposite direction,
let us denote $\theta:=\Pr[t_\A(X_n, k) > f(k) \cdot p(n)]$,
and assume there exists a function~$f$ and polynomial~$p$,
such that for all parameters~$k$ and polynomials~$q$ there is an $n_0$ such that
$\theta< f(k)/q(n)$ for all $n > n_0$. Then observe that at the time when the polynomial $q$ is chosen,
$f(k)$ is a fixed constant. Hence if $\theta< f(k)/q(n)$ holds for all polynomials~$q$, then $\theta< f(k)/\tilde{q}(n)$ also holds for the polynomial $\tilde{q}(n)$ with $\tilde{q}(n)=f(k) \cdot q(n)$, which implies $\theta< 1/q(n)$ as required by Definition~\ref{Definition: typFPT}. 

\section{\kClique is FPT on average}
\label{Section: Clique in avgFPT}

In this section we present an avgFPT-time algorithm for
the \kClique problem coupled with distribution ensembles
defined via the Erd\H{o}s-R\'{e}nyi random graph model $\cG
(n,p)$~\cite{Erdoes1}. Recall that in $\cG (n,p)$, a
random graph on the vertex set $V := \{1,\ldots, n\}$, is
constructed by connecting each pair of vertices independently
with probability $p:=p(n)$.
We will show that for any \emph{natural} function~$p$, where
the precise meaning of natural is given in
Definition~\ref{def:natural p} below, there is an avgFPT-algorithm
for \kClique
under $\cG(n,p)$, providing the first part of the proof for
Theorem~\ref{Theorem: Clique is easy}.

\begin{definition}
A function $p\colon\mathbb{N} \to [0,1]$ is natural if $p$
either equals 0 for all $n \in \mathbb{N}$, or
$p(n):=n^{-g(n)}$ for a non-negative function~$g(n)$ where the
limit $c_g:=\lim_{n\to\infty} g(n)$ exists.
\label{def:natural p}
\end{definition}

The reader should observe that most commonly used functions $p$
are natural or super-polynomially small\footnote{Note that for
super-polynomially small $p$ the \kClique problem has trivial
avgFPT and typFPT algorithms, since with super-polynomially
high probability the input graph has no edges.}. For example,
when $p(n):=1/2$ we have $g(n):=1/\lg n$ which is non-negative
and $c_g=0$, when $p(n):=1/\lg n$ we have $g(n)=\lg \lg n/\lg
n$, and for $p(n):=1/n^c$ we have $g(n)=c$.

Our proof is split into two cases,
one for dense graphs with $c_g=0$ (Section~\ref{Subsection:
Dense Average}), and the other for sparse graphs where $c_g >
0$ (Section~\ref{Subsection: Sparse Average}). Clearly, showing
that both the sparse and dense cases are in avgFPT shows that
\kClique is in avgFPT for all natural edge probabilities~$p$.

Our algorithm is very simple in both the sparse and the dense case. In the dense case, with high probability we can find a $k$-clique among a linear number of $k$-subsets of vertices. If a solution is not found amongst these vertex subsets, we can exhaustively search through all $k$-subsets of vertices in the graph since this happens with very small probability. In the sparse case, we show that the expected number of maximal cliques is polynomial, and so we can use one of many algorithms (\emph{e.g.}~\citet{Tsukiyama-et-al1977}) to compute all maximal cliques in our input.

\subsection{The dense case}
\label{Subsection: Dense Average}
Let $G \in \cG(n,p)$ where $p:=n^{-g(n)}$ with
$c_g:=\lim_{n\to\infty} g(n) = 0$, and $n$ sufficiently large.
Also, let $k \in \mathbb{N}$. Our algorithm for determining
whether $G$ has a $k$-clique, which we refer to as
algorithm~$\A$, is very simple: Let us call a clique of size
$k$ on a set of vertices $\{jk+1,\ldots,(j+1)\,k\} \subseteq V$,
for $j \in \{0,\ldots, \lfloor n/k \rfloor - 1\}$, an
\emph{elementary $k$-clique}. Algorithm~$\A$ first checks if
$G$ has an elementary $k$-clique. If so, it reports yes.
Otherwise, it tries out all $\binom{n}{k}$ subsets of $k$
vertices in $G$, reporting yes if and only if one of these is a clique.

It is clear that algorithm $\A$ correctly determines whether
$G$ has a $k$-clique in worst-case running-time $\Oh(k^2n^k)$.
Furthermore, as there are at most $\lfloor n/k
\rfloor$~elementary $k$-cliques in~$G$, checking whether
elementary $k$-cliques are present in~$G$ requires $\Oh(k^2n)$
time. Thus, if $G$ contains an elementary $k$-clique, the
running time of $\A$ is only $\Oh(k^2n)$. The next lemma shows
that for all interesting values of $k$, the probability that
this event does not occur is exponentially small.

\begin{lemma}
\label{Lemma: No Elementary Cliques}Let $k \leq \min\{n^{1/4},\, g(n)^{-1/4}\}$. Then
\[
\Prob [\,\cG (n,p) \textrm{ contains no elementary $k$-clique}\,]
\leq \exp\big(-n^{1/2}\big).
\]
\end{lemma}
\begin{proof}
Let $EK(G)$ denote the number of elementary $k$-cliques in $G$.
Observe that the probability that the vertex-subset
$\{jk+1,\ldots,(j+1)k\} \subseteq V$, for a specific $j \in
\{0,\ldots, \lfloor n/k \rfloor -1\}$, is not a $k$-clique is $1 -
p^{{k \choose 2}}$, and this probability is independent of any
other vertex-subset $\{j'k+1,\ldots,(j'+1)k\} \subseteq V$, $j'
\neq j$, being a $k$-clique. Thus, using the fact that $\lfloor
n/k \rfloor\geq n/k-1\geq n/(2k)$, we get for sufficiently
large~$n$:
\begin{equation*}
\label{eq:failure23}
\begin{split}
\Prob [EK(G)=0]
& = \left( 1 - p^{{k \choose 2}} \right)^{\lfloor n/k \rfloor}
\leq \exp \left(- \left\lfloor {n \over k}\right\rfloor p^{k \choose 2}\right)\\
&\leq \exp \left( -{n \over {2k}}\,p^{k \choose 2} \right)
= \exp \left( -{n^{1-g(n) {k \choose 2}} \over {2k}} \right).
\end{split}
\end{equation*}
Since $k \leq g(n)^{-1/4}$, we have $g(n)\binom{k}{2} \leq 1/4$
for sufficiently large $n$. Thus, since we also assume $k \leq
n^{1/4}$, the right-hand side above can be bounded by $\exp
\big( - n^{1/2} \big)$ for sufficiently large $n$.
\end{proof}

Lemma~\ref{Lemma: No Elementary Cliques} gives us an easy way
to bound the expected running-time of algorithm~$\A$. Let
$h(n):=g(n)^{-1/4}$. Observe that the worst case running-time
of algorithm $\A$ is $\Oh(k^2n^k)$. Let $h(n):=g(n)^{-1/4}$.
Then $h(n)$ tends to infinity as $n$ grows since
$\lim_{n\to\infty} g(n)^{1/4}=0$. Thus, for every $k$ there
exists a $\kappa(k)$ for which $k \leq h(n)$ for all $n \geq
\kappa(k)$. If $n < \kappa(k)$, the worst-case running time of
algorithm $\A$ can be bounded by $\Oh(k^2n^k) =
\Oh(k^2(\kappa(k))^k)$. This means that when $k > h(n) =
g(n)^{-1/4}$ (and so $n \leq \kappa(k)$), the worst-case
running-time of algorithm $\A$ can be bounded by a function in
$k$. Similarly, if $k \geq n^{1/4}$, the worst-case
running-time of $\A$ can also bounded by a function in~$k$.
Therefore, letting $f(k)$ denote a bound on the running-time of
$\A$ in case $k > \min\{n^{1/4},\, g(n)^{-1/4}\}$, we get by
Lemma~\ref{Lemma: No Elementary Cliques} above that
\begin{align*}
\E[t_\A(\cG(n,p),k)]
&= \Oh\big( f(k) + \exp(-n^{1/2}) \cdot k^2n^k +
(1-\exp(-n^{1/2})) \cdot k^2n \big)\\
&= \Oh(f(k) \cdot n),
\end{align*}
and so
\begin{equation*}
\begin{split}
\sum_{n \in \mathbb{N}} {\mathsf{E} \left[ t_\A (\cG (n,p),k) \right]\over n} = \Oh (f(k)),
\end{split}
\end{equation*}
proving that algorithm $\A$ runs in avgFPT-time.

\subsection{The sparse case}
\label{Subsection: Sparse Average}
Let $G \in \cG(n,p)$ where $p:=n^{-g(n)}$ with
$c_g:=\lim_{n\to\infty} g(n) > 0$, and let $k \in \N$. Our
algorithm for this case, which we refer to as algorithm $\B$, is
even simpler than algorithm $\A$: Algorithm $\B$ simply
computes all maximal (with respect to set inclusion) cliques in
$G$, using the classical algorithm
of~\citet{Tsukiyama-et-al1977}, and outputs yes if and only if
one of the maximal cliques is of size at least $k$. Clearly,
algorithm $\B$ correctly decides whether $G$ has a $k$-clique.

The algorithm of~\citet{Tsukiyama-et-al1977} runs in $\Oh(n^3
MK(G))$ time, where $MK(G)$ denotes the number of maximal
cliques in $G$. This is also the time complexity of algorithm
$\B$. Thus, to bound the expected running time of $\B$ on
$\cG(n,p)$, it suffices to bound the expected number of
maximal cliques that a graph in $\cG(n,p)$ contains. To ease
the analysis, we actually bound the number $K(G)$ of cliques in
$G$, for which we always have $MK(G) \leq K(G)$.

For a graph $G$ and a positive integer $s$, let $K_s (G)$
denote the number of cliques of size $s$ in $G$. For any $s\geq
2$, the expected number of cliques of size $s$ in $G\in\cG
(n,p)$ with $p=n^{-g(n)}$ is
\begin{equation}
\label{Equation: Expectancy}\begin{split}
\mu_s &:= \E \left[ K_s(\cG(n,p))\right]
= {n \choose s} p^{s\choose 2} \leq n^{s-g(n){s \choose 2}}.
\end{split}
\end{equation}
Let $s_0:= 2 \lceil {4 \over c_g} \rceil + 1$. If $n$ is
sufficiently large, then $g(n) > c_g/2$. A simple calculation
then shows that if $s \geq s_0$, then $s- g(n) {s \choose 2}
\leq s-\frac{c_g}{2}{s \choose 2} \leq -3s \leq -3$. Thereby, for any $s\geq s_0$ we have $\mu_s \leq
n^{-3}$. Using this, we can easily bound $\E \left[ K(\cG
(n,p)) \right]$ for $n$ large enough:
$$
\E\left[K(\cG(n,p))\right] = \sum_{s\geq 2} \mu_s = \sum_{s < s_0} \mu_s +
\sum_{s\geq s_0} \mu_s \leq  n^{s_0} + n\cdot n^{-3} \leq n^{s_0+1}.
$$
Hence, the expected running time of $\B$ is $\Oh(n^{s_0+4})$,
whence
\begin{equation*}
\begin{split}
\sum_{n \in \mathbb{N}} {\mathsf{E} \left[t_\B(\cG (n,p),k) \right]\over n^{s_0+4}} = \Oh (1),
\end{split}
\end{equation*}
shows that it indeed runs in avgFPT-time.

We want to point out that it is not hard to adjust the proof for the sparse case under the weaker
assumption that the limit of $g(n)$ does not  exist, but $0 < \liminf_{n \to \infty} g(n) < \limsup_{n \to \infty} g(n)$.
However, if  $0 = \liminf_{n \to \infty} g(n) < \limsup_{n \to \infty} g(n)$, then the density of the random graph varies substantially 
along appropriately chosen subsequences.
In particular, one can find a subsequence over which the random graph has very slowly decaying 
density and another subsequence in which the random graph is sparse.
In these cases, the proofs that are presented in this and the previous section
can be applied over these subsequences. Thus, effectively one could combine the two algorithms into a single algorithm. However, such an algorithm would have expected running time which is far from the expected running time that one could achieve for dense random graphs.

\section{\kClique is typically FPT}
\label{Section: Clique in typFPT}

In this section we argue that the \protect\kClique problem is
in \textnormal{typFPT} for all natural $\cG(n,p)$
distributions, completing the proof of Theorem~\ref{Theorem:
Clique is easy}. As in Section~\ref{Section: Clique in avgFPT}, our proof will split into two cases: The dense case with
$c_g=0$, and the sparse case with $c_g > 0$, where $c_g$ is the limit of the function $g(n)$ defining the edge-probability
$p:=n^{-g(n)}$. Moreover, the algorithms used in each case will be algorithms $\A$ and $\B$ of Section~\ref{Section: Clique in
avgFPT}.

Observe that Lemma~\ref{Lemma: No Elementary Cliques} shows
that in the dense case with $c_g = 0$, algorithm $\A$ runs in
$f(k) \cdot n$ time, with $f$ as given in
Section~\ref{Subsection: Dense Average}, with probability at
least $1-\exp(-n^{1/2})$. Thus, for dense edge probabilities,
algorithm $\A$ runs in typFTP-time. The main challenge here is
showing that algorithm $\B$ also runs in typFPT-time. Here,
applying a simple tail bound such as Markov's inequality, allows us to show that algorithm $\B$ is too slow with only polynomially small probability. To show that it is in fact slow only with super-polynomially small probability requires a slightly more involved argument.

So let $p:=n^{-g(n)}$ be such that
$c_g:=\lim_{n \to \infty} g(n)> 0$. Recall that the running-time of
algorithm $\B$ on a graph $G$ with $n$ vertices is $\Oh(n^3
MK(G))=\Oh(n^3 K(G))$. For an integer $s \geq 2$, we let
$K_s(G)$ denote the number of cliques of size $s$ in a graph
$G$. Then $K(G)=\sum_{s=2}^n K_s(G)$. To bound $K(\cG(n,p))$ with
high probability, we show that there exists an $s_1 \in \N$
depending only on $c_g$ (and thus on $p$) such that with very high
probability the total number of cliques of size at least $s_1$
in $\cG(n,p)$ is at most logarithmic.

\begin{lemma}
\label{lem:sizetrans}Let $p:=p(n):=n^{-g(n)}$, with $g(n)$ such that
$c_g:=\lim_{n\to\infty} g(n) >0$. Then there exists an $s_1 \in
\mathbb{N}$ such that for any $n$~sufficiently large with
probability at least $1-\exp(-n\log n)$, we have
\[
\sum_{s\geq s_1} K_s (\cG (n,p)) \leq \log n.
\]
\end{lemma}
\begin{proof}
We begin with giving a tail bound on the probability that $K_s
(\cG (n,p))$ is large for an arbitrary integer $s\geq 2$.
Recall that by~(\ref{Equation: Expectancy}), for any such
$s\geq 2$, the expected number $\mu_s$ of cliques of size $s$
is bounded, for all $s\geq 2$, by $\mu_s \leq n^{s-g(n){s \choose 2}}$
for $n$~sufficiently large. We
now give an upper-tail bound on the number of cliques of
size $s$ in $\cG(n,p)$ through which we will determine~$s$. To this end, we will use an upper-tail inequality for sums of dependent random variables due to~\citet{JanRuc05}. Let $\K$ be a non-empty set and $\{ X_S \}_{S \in \K}$ denote a
family of non-negative random variables defined on the same
probability space. For $S,S' \in \K$, we write $X_S \sim
X_{S'}$ to denote that these random variables are dependent.
For $S \in \K$, we let $\Delta_S : = |\{S' \colon  X_S' \sim
X_S \}|$ and $\Delta = \max_{S \in \K} \Delta_S$. Assume also
that for all $S \in \K$, we have $X_S \leq 1$. Now, let $X: =
\sum_{S \in \K} X_S$ and let $\mu := \E[X]$. Corollary 2.6
in~\cite{JanRuc05} states that for any $t\geq 0$,
\begin{equation}
\label{eq:UpperTail}\Prob [X \geq \mu + t] \leq
\left( 1 + {t \over \mu}\right)^{- {t \over 4 \Delta}}.
\end{equation}

\noindent
In our application, the probability space is induced
by the $\cG (n,p)$ model of random graphs and $\K$~is the
collection of all subsets of $s$ vertices of $G$. For each such
subset $S \in \K$, let $X_S \in \{0,1\}$ be the indicator
random variable which equals 1 if and only if $G[S]$~is a
clique. As far as the quantity $\Delta$~is concerned, for any
$S \in \K$ with $|S| < n$ we have
\[
\Delta_S = \sum_{i=2}^{s} {s \choose i} {n -s \choose s-i} \leq
\sum_{i=2}^s s^i n^{s-i} = n^s \sum_{i=2}^s \left({s \over n} \right)^i
\leq  n^s \sum_{i=2}^\infty \left({s \over n} \right)^i
\leq 2 s^2 n^{s-2},
\]
and therefore $\Delta \leq 2s^2 n^{s-2}$, as when $|S|=n$ then $\Delta_S =0$. Since $\sum_{S \in \K}
X_S = K_s (G) $ and letting $t= \log n/2s^2$, Inequality~\eqref{eq:UpperTail} yields
\begin{align}
\Prob \left[ K_s (G) \geq \mu_s + \log n/2s^2 \right]
&\leq \left( 1 + {\log n \over 2s^2 \mu_s} \right)^{-{\log n \over 8s^2 \Delta}} \leq
\exp \left(- {n^{g(n) {s \choose 2}} \log^2 n \over 32 n^s s^6 n^{s-2}} \right) \notag\\
&=
\exp \left(- {n^{g(n) {s \choose 2} - 2s +2} \log^2 n \over 32s^6} \right). \label{eq:Concentration}
\end{align}
Now recall that $c_g = \lim_{n \rightarrow \infty} g(n)>0$. Thus, for any
$n$~sufficiently large we have $g(n) > 8c_g/10$. Since $s\geq
2$, we also have ${s \choose 2} \geq s^2/4$, and therefore,
\[
g(n) {s \choose 2} - 2s +2 > c_g {s^2 \over 5} - 2s +2.
\]
Let us set $s_1:= \max \big\{\lceil {25 \over c_g} \rceil
,3\big\}$. We will show that for any $s\geq s_1$  we have
$c_g {s^2/ 5} -2s_0+2\geq 7$. That is, $s \left(c_g {s/5} -2
\right) \geq 3$. Indeed, $s \left(c_g {s/5} -2 \right) \geq s_1
\left( c_g {s_1/5} -2 \right)\geq s_1 \left(5-2\right) > 7$.

As $s\geq s_1 \geq 3$, for $n$~sufficiently large, this implies
that $g(n){s \choose 2} - s > s - 1\geq 2$, and therefore
$\mu_{s_1}\leq n^{-2}$. Thus if $n$~is sufficiently large, for
all $s\geq s_1$ we have
$$
\Prob \left[ K_{s} (\cG (n,p)) \geq
{\log n \over s^2} \right] \leq e^{-n \log^2n/16}.
$$
So applying the union bound we deduce that, if $n$ is
sufficiently large, with probability at least $1-e^{-n \log n}$
we have
\[
\sum_{s=s_1}^n K_{s} (\cG (n,p)) \leq
\log n \sum_{s=s_1}^{\infty} {1\over s^2} \leq \log n.
\qedhere
\]
\end{proof}
Alternatively, we could derive a weaker bound with the use 
of large deviation inequalities for subgraph statistics 
in a random graph (see for example Theorem~2.2 in~\cite{Vu01}).

The above lemma provides the existence of a
constant $s_1$ depending on $g(n)$ such that for any $n$
sufficiently large $K(\cG(n,p)) \leq n^{s_1} + \log n$ with
probability at least $1-\exp(-n \log n)$. Thus the running
time of algorithm~$\B$ on sparse graphs is $\Oh(n^{s_1+3})$
with probability at least $1-\exp(-n \log n)$, \emph{i.e.}, it
runs in typFPT-time.

\section{A Hard Distribution for \kClique}
\label{Section: Problems Hard}

In the following section we show that there exists a certain distributional ensemble for which \kClique coupled with this distribution is unlikely to have an avgFPT-algorithm, nor a typFPT-algorithm. We build on the theory developed by M{\"{u}}ller~\cite{Muller2008}, and use techniques developed in~\cite{Gurevich1991,Levin1986} and~\cite{Livne2006} to prove our argument.

We begin by defining our average-case analog of W[1]. A distribution ensemble $X$ is said to be \emph{simple}\footnote{M{\"{u}}ller~\cite{Muller2008} uses here the term \emph{polynomial-time distributed}} if there is a polynomial algorithm that on input $x \in \{0,1\}^*$, outputs the probability $\textsf{Pr}[X_{|x|} \leq x]$, where $\leq$ denotes the standard lexicographic order on strings. In the classical world, the standard definition of the average-case analog of NP is defined as all NP problems coupled with simple distributions. The restriction to simple distributions is done in order to avoid trivial hardness results. Thus, adapting the same line of discourse to the parameterized world, we define the class distW[1] as the set
$$
\textnormal{distW[1]}:=\big\{(L,X) \colon X \textrm{ is a simple distribution ensemble
and } L \in \textnormal{W[1]}\big\}.
$$
Note that this definition easily extends to any other parameterized class besides W[1]. The main working conjecture we propose for average-case parameterized analysis is $\textnormal{distW[1]} \nsubseteq \textnormal{avgFPT} \cup \textnormal{typFPT}$.

We next define a reduction that preserves average-case parameterized tractability. The notion of a reduction we use here is essentially a hybrid of the two corresponding notions in classical average-case complexity and parameterized complexity.
\begin{definition}
\label{Definition: Reduction}A distributional parameterized problem $(L_1,X)$ \emph{reduces}
to another distributional parameterized problem $(L_2,Y)$, if
there exists an algorithm~$\A$, a
function $f$, and
a polynomial $p$, such that $\A$ on input $(x,k) \in \Sigma^*
\times \mathbb{N}$ outputs in time $f(k) \cdot p(|x|)$ a pair
$(y,\ell) \in \Sigma^* \times \mathbb{N}$ satisfying:
\begin{itemize}
\item $(x,k) \in L_1 \iff (y,\ell) \in L_2$.
\item $\ell \leq f(k)$.
\item $|x|\leq |y|$.
\item $\Prob[\A(X_{|x|},k)=(y,\ell)] \,\leq\, f(k) \cdot
    p(|x|) \cdot \mathsf{Pr}[Y_{|y|}=y]$.
\end{itemize}
\end{definition}

Observe that the first two requirements in Definition~\ref{Definition: Reduction} are the usual requirements of a parameterized reduction. The third requirement is a technical requirement used also in non-parameterized distributional reductions that can typically be satisfied by a straightforward padding argument, yet it is necessary for the composition of our reductions (see Lemma~\ref{Lemma: Transitivity}). We note that this requirement is missing in M{\"{u}}ller's work~\cite{Muller2008} since he was not interested in composing reductions. The last requirement, often referred to as the \emph{domination property}, ensures that an infrequent input of $L_1$ does not get mapped to a frequent input of $L_2$. We let $(L_1,X) \leq (L_2,Y)$ denote the fact that $(L_1,X)$ reduces, as per Definition~\ref{Definition: Reduction}, to $(L_2,X)$.

\begin{lemma}
\label{Lemma: Transitivity}$\leq$ is transitive.
\end{lemma}
\begin{proof}

Let $(L_1,X)$, $(L_2,Y)$, and $(L_3,Z)$ be three distributional parameterized problems with $(L_1,X) \leq (L_2,Y)$ and $(L_2,Y) \leq (L_3,Z)$, and let $\A_1$ and $\A_2$ respectively be the algorithms showing that $(L_1,X) \leq (L_2,Y)$ and $(L_2,Y) \leq (L_3,Z)$, as required by Definition~\ref{Definition: Reduction}. We prove that $(L_1,X) \leq (L_3,Z)$, by showing that the composition of $\A_2$ and $\A_1$ gives an algorithm that satisfies the conditions of Definition~\ref{Definition: Reduction}. It is easy to verify that the first three requirements of of Definition~\ref{Definition: Reduction} hold. In particular, for any $(x,k) \in \Sigma^* \times \mathbb{N}$, the running-time of $\A_2(\A_1(x,k))$ (and hence, also its output size) is bounded by $f(k) \cdot p(|x|)$ for some computable $f()$ and polynomial $p()$, and moreover we have $m \leq f(k)$. To prove the lemma, we show that the probability that $\A_2(\A_1(X_{|x|},k))$ outputs and $(z,m)  \in \Sigma^* \times \mathbb{N}$ is bounded by above by the probability of $z$ according to $Z_{|z|}$, modulo some FPT-factor in $|x|$ and $k$.

For this, note that all four requirements of Definition~\ref{Definition: Reduction} for $\A_1$ and $\A_2$ hold with $f()$ and $p()$, and write
$$
\Prob[\A_2(\A_1(X_{|x|},k))=(z,m)] \quad = \quad \sum_{\ell} \,\,\,\,\sum_{n} \!\! \sum_{\substack{y \text{ s.t } |y|=n,\\\A_2(y,\ell)=(z,m)}} \!\!\!\!\!\!\!\!\Prob[\A_1(X_{|x|},k))=(y,\ell)].
$$
Let $n^*$ and $\ell^*$ denote the values of $n$ and $\ell$ that maximize the rightmost sum above. Since there are only $f(k) \cdot p(|x|)$ choices for pairs $(\ell,n)$, we can restrict ourselves to bounding the rightmost sum above in terms of $n^*$ and $\ell^*$. By definition of $\A_1$, we have
$$
\sum_{\substack{y^* \text{ s.t. } |y^*|=n^*,\\\A_2(y^*,\ell^*)=(z,m)}} \!\!\!\!\!\!\!\!
\Prob[\A_1(X_{|x|},k))=(y,\ell^*)] \,\,\, \leq \,\,\,
f(k) \cdot p(|x|) \,\,\, \cdot \!\!\!\!\!\!\!\!\!\! \sum_{\substack{y \text{ s.t. } |y|=n^*,\\\A_2(y,\ell^*)=(z,m)}}
\!\!\!\!\!\!\!\! \Prob[Y_{n^*}=y].
$$
Thus it suffices to bound the sum of probabilities in the rightmost sum above. Observe that this sum is precisely the probability that $\A_2(Y_{n^*},\ell^*)=(z,m)$. By definition of $\A_2$, we get that
$$
\Prob[\A_2(Y_{n^*},\ell^*)=(z,m)] \,\,\, \leq \,\,\, f(\ell^*) \cdot p(n^*) \,\, \cdot \,\, \Prob[Z_{|z|}=z].
$$
Now, recall that $\ell^* \leq f(k)$, and that $n^*=|y| \leq |z| \leq f(k) \cdot p(|x|)$ for
every $y$ as above (by the third requirement of Definition~\ref{Definition: Reduction}). Thus,
$f(\ell^*) \cdot p(n^*) \leq  f(f(k))p(f(k)) \cdot p(p(|x|))$, and the lemma is proven. 

\end{proof}

The next lemma shows the most important property of our reductions: For any pair of distributional parameterized problems $(L_1,X)$ and $(L_2,Y)$ with $(L_1,X) \leq L_2,Y)$, the question of whether $(L_1,X)$ is tractable in the average-case parameterized sense reduces to same question regarding $(L_2,Y)$. This has been shown for avgFPT-algorithms by M{\"{u}}ller~\cite{Muller2008}\footnote{In fact, \cite{Muller2008} shows this for a more relaxed notion of reduction where the third requirement does not exist.}. We complement this result by showing that the same holds for typFPT-algorithms. For completeness, we also provide a proof for avgFPT in the appendix of the paper.

\begin{lemma}
\label{Lemma: DistFPT Preservation}If $(L_1,X) \leq (L_2,Y)$ and $(L_2,Y)$ has a \textnormal{typFPT}-algorithm, then $(L_1,X)$ also has a \textnormal{typFPT}-algorithm.
\end{lemma}
\begin{proof}

Let $\A$ be a typFPT algorithm for $(L_2,Y)$ running in $f_\A(\ell)\cdot p_\A(|y|)$ time with high probability, and let $\R$ denote a reduction from $(L_1,X)$ to $(L_2,Y)$, as required by Definition~\ref{Definition: Reduction}, running in $f_\R(k)\cdot p_\R(|x|)$ time. We argue that the algorithm $\B$ which outputs $\B(x,k) := \A(\R(x,k))$ for all $(x,k) \in \Sigma^* \times \mathbb{N}$ is a typFPT-algorithm for $(L_1,X)$. By definitions of $\R$ and $\A$, it is clear that $\B$ correctly decides $(L_1,X)$. We show that algorithm $\B$ runs in more than $f(k) \cdot p(|x|)$ time with super-polynomially small probability, for $f()$ and $p()$ chosen such that $f(k) \cdot p(n)$ is sufficiently larger than $f_\R(k)\cdot p_\R(n) + f_\A(k)\cdot p_\R(n)$ for all $k$ and sufficiently large $n$.

Fix $k \in \mathbb{N}$, and let $q()$ be an arbitrary polynomial. By our choice of $f()$ and $p()$, we can bound the the probability that $\B$ runs in more than $f(k) \cdot p(|x|)$ time by
$$
\Prob[t_\B(X_{|x|},k) > f(k) \cdot p(|x|)] \,\,\,\, \leq \,\,\,\,
\sum_\ell \,\, \sum_{n} \!\!\!\!\!\! \sum_{\substack{y \text{ s.t. } |y|=n,\\ t_\A(y,\ell) \,>\, f_\A(\ell) \cdot p_\A(n)}} \!\!\!\!\!\!\!\!\!\! \Prob[\R(X_{|x|},k)=(y,\ell)].
$$
Note that there are at most $f_\R(k) \cdot p_\R(|x|)$ pairs of $(\ell,n)$ in the righthand side above. Thus, we can bound the total summation on the righthand side in terms of $\ell^*$ and $n^*$ which are the values of $\ell$ and $n$ that maximize the rightmost sum in this summation. Due to the requirements on $\R$, we get
$$
\sum_{\substack{y \text{ s.t. } |y|=n^*,\\ t_\A(y,\ell^*) \,>\, f_\A(\ell^*) \cdot p_\A(n^*)}} \!\!\!\!\!\!\!\!\!\!\!\!\!\!\!\!
\Prob[\R(X_{|x|},k)=(y,\ell^*)] \,\,\,\, \leq \,\,\,\, f_\R(k) \cdot p_\R(|x|) \,\,\,\cdot \!\!\!\!\!\!\!\!\!\!\!\!\!\!\!\!
\sum_{\substack{y \text{ s.t. } |y|=n^*,\\ t_\A(y,\ell^*) \,>\, f_\A(\ell^*) \cdot p_\A(n^*)}} \!\!\!\!\!\!\!\!\!\!\!\!\!\!\!\! \Prob[Y_{n^*}=y].
$$
Note that the rightmost sum is just the probability that $\A(Y_{n^*},\ell^*)$ runs in more than $f_\A(\ell^*) \cdot p_\A(n^*)$ time. Since $\A$ is a typFPT-algorithm for $(L_2,Y)$, this probability is super-polynomially small. In particular, it smaller than $1/q'(n)$, where $q'(n):=(f_\R(k) \cdot p_\R(n))^2 \cdot q(n)$. Note that $q'(n)$ is indeed a polynomial, as $p_\R()$ and $q()$ are polynomials, and $f(k)$ is fixed. Thus, we have
\begin{eqnarray*}
\Prob[t_\B(X_{|x|},k) > f(k) \cdot p(|x|)]  \,\,\,\,\leq \\
(f_\R(k) \cdot p_\R(|x|))^2 \,\, \cdot \,\, \Prob[t_\A(Y_{n^*},\ell^*) >f_\A(\ell^*) \cdot p_\A(n^*)]  \,\,\,\,\leq \\
\frac{(f_\R(k) \cdot p_\R(|x|))^2}{q'(|x|)} \,\,\,\, & = \,\,\,\, \frac{1}{q(|x|)},
\end{eqnarray*}
and the lemma is proven. 

\end{proof}

By \emph{\textnormal{distW[1]}-complete} we will mean, as usual, a problem $(L,X) \in \textnormal{distW[1]}$ with $(L',Y) \leq (L,X)$ for every problem $(L',Y)$ in distW[1]. Note that an avgFPT algorithm or a typFPT algorithm for a distW[1]-complete problem would falsify our working conjecture of $\textnormal{distW[1]} \nsubseteq \textnormal{avgFPT} \cup \textnormal{typFPT}$. We therefore argue that showing that a problem is distW[1]-complete is strong evidence against the existence of such algorithms. In the remainder of the section we prove the following theorem:
\begin{theorem}
\label{Theorem: Clique is hard}Let $L$ denote the \kClique problem. There exists a simple distribution $Y$ for which $(L,Y)$ is \textnormal{distW[1]}-complete.
\end{theorem}

For proving Theorem~\ref{Theorem: Clique is hard}, we need two initial results. The first states that there exists some (artificial) distW[1]-complete problem. This has been shown by M{\"{u}}ller~\cite{Muller2008} using the same ideas as in~\cite{Gurevich1991,Levin1986}. While M{\"{u}}ller uses a slightly different notion of reduction than ours (his definition lacks the third requirement of Definition~\ref{Definition: Reduction}), his proof can easily be adopted to accommodate also our definition by a straightforward padding argument.

\begin{theorem}[\cite{Muller2008}]
\label{Theorem: distW[1]-complete}There is a distributional parameterized problem $(U,X)$ which is \textnormal{distW[1]}-complete.
\end{theorem}

The following lemma by \citet{Livne2006} (see also~\cite{Goldreich}) gives the necessary technical tool for reducing the $(U,X)$ problem above to some distributional \kClique. We assume some natural encoding of graphs into binary strings, and let $\langle G \rangle$ denote the encoding of a given graph $G$.
\begin{lemma}[\cite{Livne2006}]
\label{Lemma: Livne transformation}There is a polynomial-time algorithm that given a graph~$G$ and an $x \in \{0,1\}^*$, computes a graph~$G_x$ such that:
\begin{itemize}
\item $x=x'$ and $G=G'$ $\iff$ $\langle G_x \rangle = \langle G_{x'} \rangle$.
\item $|x|=|x'| \iff |\langle G_x \rangle| = |\langle G_{x'} \rangle|$.
\item $|x| \leq |\langle G_x \rangle|$.
\item $G$ has a $k$\nobreakdash-clique $\iff$ $G_x$ has a $k$\nobreakdash-clique, for any $k \neq 2$.
\item If $X$ is a simple distribution ensemble then the distribution ensemble $Y$ defined by
\[
\Prob[Y_{|y|} = y] = \left\{
\begin{array}{lll}
\Prob[X_{|x|}=x] & : & \,\,  y= \langle G_x \rangle \\
0 & : & \,\, y \neq \langle G_x \rangle \text{ for all } x \text{ and } \exists x \text{ s.t. }\langle G_x \rangle \in \{0,1\}^{|y|}\\
1/2^{|y|} & : & \,\, \text{otherwise } (\nexists x\text{ s.t. }\langle G_x \rangle \in \{0,1\}^{|y|})
\end{array}
\right.
 \]
is also simple.
\end{itemize}
\end{lemma}

\begin{proof}[Proof of Theorem~\ref{Theorem: Clique is hard}]
Let $(U,X)$ denote the distW[1]-complete problem of
Theorem~\ref{Theorem: distW[1]-complete}, and let $L$ denote
the \kClique problem. Since $U \in \textnormal{W[1]}$, and
$L$~is W[1]-complete, there exists a parameterized reduction
$\A$ from $U$ to $L$. We construct an alternative reduction
$\A^*$ which works as follows:
\begin{enumerate}
\item It first computes $\A(x,k) = (G,\ell)$.
\item It then checks if $\ell = 2$:
    \begin{itemize}
    \item[$(a)$] If so, it sets $\ell^*:=3$ if $G$ has
        no edges, and otherwise it sets $\ell^*:=1$.
    \item[$(b)$] If $\ell \neq 2$, it sets
        $\ell^*:=\ell$.
    \end{itemize}
\item It then computes $G_x$, and outputs the pair
    $(G_x,\ell^*)$.
\end{enumerate}

\noindent
Clearly, $\A^*$ runs in FPT-time. Moreover, $\A^*$~is a
reduction, as required by Definition~\ref{Definition:
Reduction}, from $(U,X)$ to $(L,Y)$, where $Y$~is the
distribution defined in the last item of Lemma~\ref{Lemma:
Livne transformation} above. Indeed, it is easy to see that
\[
(x,k) \in U \iff (G,\ell) \in L \iff (G_x,\ell^*) \in L
\]
by Lemma~\ref{Lemma: Livne transformation} and the definition
of $\A$. Furthermore, since $\ell \leq f(k)$ for some
$f$, we have $\ell^* \leq f(k) + 1$, and $|x| \leq
|\langle G_x \rangle|$ by Lemma~\ref{Lemma: Livne
transformation}. Finally, by our construction and
Lemma~\ref{Lemma: Livne transformation},
$$
\Prob[\A^*(X_{|x|},k)=(G_x,\ell^*)] = \Prob[X_{|x|}=x] =
\Prob[Y_{|\langle G_x \rangle|}=\langle G_x \rangle].
$$
Thus $(U,X) \leq (L,Y)$. Since $Y$~is simple, $(L,Y) \in
\textnormal{distW[1]}$, and so by Lemma~\ref{Lemma:
Transitivity} we get that $(L,Y)$ is distW[1]-complete.
\end{proof}

\section{Discussion}
\label{Discussion}

In this paper we considered the average-case parameterized complexity of the fundamental \kClique problem. We showed that when restricted to Erd\H{o}s-R\'{e}nyi random graphs of arbitrary density $p:=p(n)$, the problem admits two types of natural average-case analogs of FPT algorithms: An avgFPT algorithm and a typFPT algorithm. Thus, in this sense, the worst-case W[1]-complete \kClique problem is easy on average. Furthermore, by adaptation of arguments from classical average-case analysis due to~\citet{Livne2006}, it can also be shown that for specific distributions \kClique is unlikely to be FPT on average (unless any problem in W[1] under any computable distribution is easy).

It would be interesting to see which other distributions make \kClique easy~\cite{FriedrichKrohmer12} and which other W[1]-hard problems are easy on Erd\H{o}s-R\'{e}nyi random graphs of arbitrary density $p:=p(n)$. Here it important to require that the algorithms are deterministic and always correct, to avoid trivial results. We remark that many of the arguments used for \kClique do not seem to carry through easily to other problems. A particularly interesting case is the \textsc{$k$-Dominating Set} problem, the problem of determining whether a given graph has a dominating set of size $k$. The hard instances for this problem seem to be $\cG(n,1/2)$.

\bibliographystyle{abbrvnat}

\newpage
\appendix
\section{Appendix}

In this section we provide proofs for claims used in Section~\ref{Section: Problems Hard} which are proven in M{\"{u}}ller's thesis~\cite{Muller2008} for definitions which are slightly different then ours. In particular we provide a proof for the avgFPT analog for Lemma~\ref{Lemma: DistFPT Preservation}, and a proof for Theorem~\ref{Theorem: distW[1]-complete}. Our proofs here use the same techniques as in~\cite{Muller2008}.

\begin{lemma}
If $(L_1,X) \leq (L_2,Y)$ and $(L_2,Y)$ has an \textnormal{avgFPT}-algorithm, then $(L_1,X)$ also has an \textnormal{avgFPT}-algorithm.
\end{lemma}
\begin{proof}

Let $\A$ be the algorithm as in Definition~\ref{Definition:
avgFPT} showing that $(L_2,Y) \in \textnormal{avgFPT}$, and let
$\R$ denote the reduction from $(L_1,X)$ to $(L_2,Y)$, as
required by Definition~\ref{Definition: Reduction}. Also, let
$f_\A$ and $p_\A$ be the computable function and polynomial
associated with $\A$, and let $f_\R$ and $p_\R$ be the
computable function and polynomial associated with $\R$. We
show that the algorithm $\B$ which outputs $\B(x,k) :=
\A(\R(x,k))$ for all $(x,k) \in \Sigma^* \times \mathbb{N}$
gives a avgFPT algorithm for $(L_1,X)$.

By definitions of $\R$ and $\A$, it is clear that $\B$
correctly decides $(L_1,X)$. Furthermore, since for any $(x,k)
\in \Sigma^* \times \mathbb{N}$, we have
$t_\B(x,k)=t_\R(x,k)+t_\A(\R(x,k))+O(1)$, by linearity of
expectation, we have
$$
\sum_{n \in \mathbb{N}} \E\left[\frac{t_\B(X_n,k)}{n^c}\right] \leq
\sum_{n \in \mathbb{N}} \E\left[\frac{t_\R(X_n,k)}{n^c}\right] +
\sum_{n \in \mathbb{N}} \E\left[\frac{t_\A(\R(X_n,k))}{n^c}\right],
$$
for any $c \in \mathbb{N}$. As $t_\R(x,k) \leq f_\R(k) \cdot
p_\R(|x|)$ for all $(x,k) \in \Sigma^* \times \mathbb{N}$, we
have for any $k \in \mathbb{N}$
$$
\sum_{n \in \mathbb{N}}
\E\left[\frac{t_\R(X_n,k)}{n^c}\right] = \Oh \left( f_\R(k) \right)
$$
for some sufficiently large $c$. Thus, to prove the lemma it
suffices to bound the second summation above for every $k \in
\mathbb{N}$.

Fix $k \in \mathbb{N}$. Due to the requirements on $\R$, we
have for every $x \in \Sigma^*$
$$
\begin{array}{ll}
\E[t_\A(\R(X_{|x|},k))] & = \sum_\ell \sum_{n} \sum_{|y|=n}
t_\A(y,\ell) \cdot \Prob[\R(X_{|x|},k)=(y,\ell)] \\
& \leq f_\R(k) \cdot p_\R(|x|) \cdot \sum_\ell \sum_{n} \sum_{|y|=n}
t_\A(y,\ell) \cdot \Prob[Y_{n}=y] \\
& = f_\R(k) \cdot p_\R(|x|) \cdot \sum_\ell \sum_{n} \E [t_\A (Y_{n},\ell)].
\end{array}
$$
Now observe, that the number of summands on the right-hand side
of the above inequality is finite, and, therefore, there exist
$n^*,\ell^*$ that maximize the summands. In particular, observe
that the number of summands is at most $f_\A (k) \cdot p_\A (n)$.
Thus,
$$ \E [t_\A (\R (X_{|x|},k))] \leq f_\R (k) \cdot f_\A (k) \cdot p_{\R}(|x|) \cdot p_\A (|x|) \cdot \E [t_\A (Y_{n^*},\ell^*)].$$
But $ n^* \leq f_\A (k) \cdot p_\A (n) $, which, in turn, implies
that for any $c>0$ we have
$$ (n^*)^c \leq \left(f_\A (k) \cdot p_\A (|x|) \right)^c |x|^c .$$
Thus, for any positive $c$ we have
$$ \E \left[ {t_\A (\R (X_{|x|},k)) \over |x|^c} \right] \leq {f_\R (k) \cdot f_\A (k) \over f_\A^c (k)}~{p_{\R}(|x|) \cdot p_\A (|x|) \over p_\A^c (|x|)}~
\E \left[{t_\A (Y_{n^*},\ell^*) \over (n^*)^c} \right].$$ As we need to take
the sum of the above over all $n\in \mathbb{N}$, observe that
on the right-hand side the same value of $n^*$ can be repeated
at most $n^*$ times. Thus, we obtain
\begin{eqnarray*}\sum_{n \in \mathbb{N}} \E \left[{t_\A (\R (X_{n},k))\over n^c}\right] & \leq {f_\R (k)
f_\A^c (k) \over f_\A^c (k)}~\sum_{n \in \mathbb{N}} n~{p_{\R}(n)
p_\A (n) \over p_\A^c (n)}~ \E \left[{t_\A (Y_{n},\ell^*(n)) \over
n^c}\right] \nonumber \\
& \leq {f_\R (k) \over f_\A^{c-1} (k)}~\sum_{n \in \mathbb{N}} n~{p_{\R}(n)
 \over p_\A^{c-1} (n)}~ \E \left[{t_\A (Y_{n},f_\A (k))] \over
n^c}\right].
\end{eqnarray*}
Choosing $c$ large enough, concludes the proof of the lemma.

\end{proof}

Before providing the proof of Theorem~\ref{Theorem: distW[1]-complete}, we need to describe the machine characterization for W[1] of \citet{ChenFlumGrohe03}. The characterization is based on a nondeterministic version of \emph{random access machines} (RAM) which are a more accurate model of real-life computation than Turing machines. A RAM consists of an infinite set of registers $\{r_0,r_1,r_2,\ldots\}$, a program counter $x$, and an instruction set. The instructions are of the form \verb"STORE"~$i$ or \verb"ADD"~$i,j$, and so forth (see~\cite{ChenFlumGrohe03} for details). A \emph{nondeterministic RAM} (NRAM) consists of an additional instruction of the form \verb"GUESS" $i,j$, which results in the machine ``guessing" a number less than or equal to the number stored in register $r_i$, and storing this number in $r_j$~\cite{ChenFlumGrohe03}. \citeauthor{ChenFlumGrohe03} used the following type of NRAM programs to characterize W[1]:
\begin{definition}
A \textnormal{NRAM} program $P$ is a \textnormal{W[1]}-program if there exists a computable function $f$ and a polynomial $p$ such that on every input $(x,k)$, the program $P$ on every run
\begin{itemize}
\item performs at most $f(k) \cdot p(|x|)$ instructions, storing numbers which are $\leq f(k) \cdot p(|x|)$ only in the first $f(k) \cdot p(|x|)$ registers;
\item in every run of $P$, all nondeterministic instructions are among the last $f(k)$ instructions of the computation.
\end{itemize}
In this case, we say that $P$ accepts $(x,k)$ using $(f(k),p(|x|))$ resources.
\end{definition}

\begin{theorem}[\cite{ChenFlumGrohe03}]
\label{Theorem: W1-Program}A parameterized problem $L$ is in \textnormal{W[1]} iff there exists a \textnormal{W[1]}-program $P$ deciding $L$.
\end{theorem}

Theorem~\ref{Theorem: W1-Program} above suggests the following universal problem $U$ for W[1]: Given an NRAM program $P$, an input $(x,\ell) \in \Sigma^*$, a unary integer $t$, and a parameter $k$, decide whether $P$ accepts $(x,\ell)$ using $(t,k)$ resources. It is clear that $U$~is in W[1]: On input $(\langle P,(x,\ell),t\rangle, k)$, a W[1]-program $Q$ can simulate, using $(\Oh(t),\Oh(k))$ resources, all runs of $P$ on $(x,\ell)$ that use $(t,k)$ resources. We next define a simple uniform distribution ensemble $Y$ for $U$ given by
$$
\mathsf{Pr}[Y_n=\langle P,(x,\ell),t\rangle] := \frac{1}{2^{|P|+|x|} \cdot (\ell+t)},
$$
where $n:=|P|+|x|+\ell+t$. It is not difficult to verify that under a suitable encoding of NRAM programs, the above distribution is simple. Thus, $(U,Y) \in \textnormal{distW[1]}$. We will show that $(U,Y)$ is in fact distW[1]-complete, using the following lemma initially proved by~\citet{Levin1986}.
\begin{lemma}[\cite{Levin1986}]
\label{Lemma: Smoothing}Let $X$ be a simple distribution ensemble. Then there exists a polynomial-time computable, and polynomial-time invertible, injective function $\Psi\colon\Sigma^* \to \Sigma^*$, such that for all $x \in \Sigma^*$ we have $\mathsf{Pr}[X_{|x|}=x] \leq 2^{-(|\Psi(x)|+1)}$.
\end{lemma}

\begin{proof}[Proof of Theorem~\ref{Theorem: distW[1]-complete}]
Let $(L,X)$ be a problem in distW[1]. We reduce $(L,X)$ to $(U,Y)$ by mapping an instance $(x,k) \in \{0,1\}^* \times \mathbb{N}$ to an instance $(\langle P,(x',k),t \rangle,\ell)$ as follows: Denote by $\Psi$ the function given in Lemma~\ref{Lemma: Smoothing}, and let $p_\Psi$ be the polynomial bounding the running-time of computing and inverting $\Psi$. Since $(L,X) \in \textnormal{distW[1]}$, $L \in \textnormal{W[1]}$, and so by Theorem~\ref{Theorem: W1-Program} there is a W[1]-program $Q$ deciding $L$. Let $f_Q$ and $p_Q$ denote the computable function and polynomial associated with $Q$ as in 
Theorem~\ref{Theorem: W1-Program}. Define $P$ to be the program that gets $x':=\Psi(x)$ as input, computes $x=\Psi^{-1}(x')$, and then simulates $Q$ on $(x,k)$ (accepting iff $Q$ accepts). Finally, define $t:=p_\Psi(|x'|)+p_Q(|x|+\ell) + c$, where $c$~is the overhead time required to simulate $\Psi^{-1}$ and $Q$, and let $\ell:=f_Q(k)$.

Observe that our construction can be carried out in FPT-time, since writing down $P$~is done in time independent of $(x,k)$. Furthermore, clearly $\ell \leq f_Q(k)$, and since $Q$ decides $L$, we have $(x,k) \in L \iff (\langle P,(x',k),t \rangle,\ell) \in U$. Thus, the first two requirements of Definition~\ref{Definition: Reduction} are satisfied by the construction. The third requirement can be satisfied by padding $P$ as necessary. Finally, to see that the last requirement is also satisfied, observe that the probability of $y:=\langle P,(x',k),t \rangle$ in $Y$ is at least 
$$
\mathsf{Pr}[Y_{|y|}=y] := \frac{1}{2^{|P|+|\Psi(x)|} \cdot (k+t)} \geq \frac{1}{c' \cdot |y|} \cdot \frac{1}{2^{|\Psi(x)|}},
$$
where $c'$~is a constant depending only on $P$ and $\Psi$, and not on $(x,k)$. On the other hand, according to Lemma~\ref{Lemma: Smoothing} we have
$$
\mathsf{Pr}[X_{|x|}=x] \leq \frac{1}{2^{|\Psi(x)|+1}}.
$$
Thus, by letting $p$ denote the polynomial $p(n):=c'n/2$, combining these two inequalities gives
$$
\mathsf{Pr}[X_{|x|}=x] \leq p(|y|) \cdot \mathsf{Pr}[Y_{|y|}=y].
$$
Noting that $(x,k)$~is the only pair that gets mapped to $(y,\ell)$ by our construction, the theorem follows.
\end{proof}

\end{document}